\documentclass[12]{article}
\usepackage{amsmath,amsthm}

\usepackage{amsmath}
\usepackage{amsfonts}
\usepackage{amssymb}

\usepackage[usenames,dvipsnames]{color}
\usepackage[noend]{algorithm, algorithmic}
\floatname{algorithm}{Algorithm}
\usepackage{url}
\usepackage{xspace}

\newtheorem{whatever}{Whatever}
\newtheorem{lemma}[whatever]{Lemma}
\newtheorem{theorem}[whatever]{Theorem}

\newtheorem{proposition}[whatever]{Proposition}

\newtheorem{definition}[whatever]{Definition}

\newcommand{\comment}[1]{}

\newcommand{\BAS}{\begin{align*}}
\newcommand{\EAS}{\end{align*}}
\newcommand{\BA}{\begin{align}}
\newcommand{\EA}{\end{align}}

\newcommand{\ph}{P(H)}
\newcommand{\pl}{P(L)}

\newcommand{\al}{\alpha}

\newcommand{\ie}{{\it i.e.}}
\newcommand{\eg}{{\it e.g.}}
\newcommand{\ones}{{\mathbf 1}}
\newcommand{\zeros}{{\mathbf 0}}
\newcommand{\BEQ}{\begin{equation}}
\newcommand{\EEQ}{\end{equation}}
\newcommand{\BIT}{\begin{itemize}}
\newcommand{\EIT}{\end{itemize}}
\newcommand{\BNUM}{\begin{enumerate}}
\newcommand{\ENUM}{\end{enumerate}}

\def\ph{\mathcal{P}[H]}
\def\pl{\mathcal{P}[L]}
\def\sph{p[H]}
\def\spl{p[L]}
\def\sqh{q[H]}

\def\sqjh{q_j[H]}

\def\pib{\bar{p}_{i}}
\def\sphib{\bar{p}_{i}[H]}

\def\sphipj{p_{\ipj}[H]}

\def\M{\mathcal{M}}
\def\eps{\epsilon}
\def\truth{X}
\def\invert{X^c}

\def\ipj{r_{j}(i)}
\def\ipjl{r_{j_l}(i)}

\title{Crowdsourced Judgement Elicitation with Endogenous Proficiency}

\author{
Anirban Dasgupta\thanks{Yahoo! Labs, Sunnyvale, CA 95054. Email: \texttt{anirban@yahoo-inc.com}}\and
Arpita Ghosh\thanks{Cornell University, Ithaca, NY 14853. Email: \texttt{arpitaghosh@cornell.edu}}
}
\date{}

\begin{document}

\maketitle

\begin{abstract}
Crowdsourcing is now widely used to replace judgement or evaluation by an expert authority with an aggregate evaluation from a number of non-experts, in applications ranging from rating and categorizing online content all the way to evaluation of student assignments in massively open online courses (MOOCs) via peer grading. A key issue in these settings, where direct monitoring of both effort and accuracy is infeasible, is incentivizing agents in the `crowd' to put in {\em effort} to make good evaluations, as well as to truthfully report their evaluations. We study the design of mechanisms for crowdsourced judgement elicitation when workers strategically choose {\em both} their reports and the effort they put into their evaluations. This leads to a new family of information elicitation problems with unobservable ground truth, where an agent's proficiency--- the probability with which she correctly evaluates the underlying ground truth--- is {\em endogenously} determined by her strategic choice of how much effort to put into the task.

Our main contribution is a simple, new, mechanism for binary information elicitation for multiple tasks when agents have {\em endogenous proficiencies}, with the following properties: (i) Exerting maximum effort followed by truthful reporting of observations is a Nash equilibrium. (ii) This is the equilibrium with {\em maximum payoff} to all agents, {\em even} when agents have different maximum proficiencies, can use mixed strategies, and can choose a different strategy for each of their tasks. Our information elicitation mechanism requires only minimal bounds on the priors, asks agents to only report their own evaluations, and does not require any conditions on a diverging number of agent reports per task to achieve its incentive properties. The main idea behind our mechanism is to use the presence of {\em multiple} tasks and ratings to estimate a reporting statistic to identify and penalize low-effort agreement--- the mechanism rewards agents for agreeing with another `reference' agent report on the same task but also penalizes for {\em blind agreement} by subtracting out this statistic term, designed so that agents obtain rewards {\em only} when they put in effort into their observations.
\end{abstract}

\section{Introduction}
\label{s-intro}
%Introduce crowdsourcing; judgement elicitation applications
Crowdsourcing, where a problem or task is broadcast to a crowd of potential participants for solution, is used for an increasingly wide variety of tasks on the Web. One particularly common application of crowdsourcing is in the context of making evaluations, or judgements--- when the number of evaluations required is too large for a single expert, a solution is to replace the expert by an evaluation aggregated from a `crowd' recruited on an online crowdsourcing platform such as Amazon Mechanical Turk. Crowdsourced judgement elicitation is now used for a number of applications such as image classification and labeling, judging the quality of online content, identifying abusive or adult content, and most recently for peer grading in online education, where Massively Open Online Courses (MOOCs) with enrollment in the hundreds of thousands crowdsource the problem of evaluating assignments submitted by students back to the class itself. While one issue in the context of crowdsourcing evaluations is how best to aggregate the evaluations obtained from the crowd, there is also a key question of {\em eliciting the best possible evaluations} from the crowd in the first place.

The problem of designing incentive mechanisms for such crowdsourced judgement elicitation scenarios has two aspects. First, suppose each worker has already evaluated, or formed a judgement on, the tasks allocated to her. Since the `ground truth' for each task is unknown to the system, a natural solution is to reward workers based on other workers' reports for the same task (this being the only available source of information about this ground truth)\footnote{It is of course infeasible for a requester to monitor every worker's performance on her task, since this would be a problem of the same scale as simply performing all the tasks herself. We also note that a naive approach of randomly checking some subset of evaluations, either via inserting tasks with known responses, or via random checking by the requester, turns out to be very wasteful of effort at the scale neccessary to achieve the right incentives.}. The problem of designing rewards to incentivize agents to truthfully report their observation, rather than, for example, a report that is more likely to agree with other agents' reports, is an {\em information elicitation} problem with unobservable ground truth. Information elicitation has been recently been addressed in the literature in the context of eliciting opinions online (such as user opinions about products, or experiences with service providers); see \S \ref{s-relwork}. However, in those settings, agents (users) have {\em already} formed an opinion after receiving their signal (for example a user who buys a product forms an opinion about it after buying it)--- so agents only need to be incentivized to incur the cost to report this opinion and find it more profitable to report their opinions {\em truthfully} than to report a different opinion.

%pre-formed --> experientially formeds
In the crowdsourcing settings we consider, however, the user does not have such a pre-formed, or experiential, opinion anyway, but rather forms a judgement {\em as part of her task}--- further, the accuracy of this judgement depends on whether or not the agent puts in effort into it (for instance, a worker evaluating whether images contain objectionable content could put in no effort and declare all images to be clean, or put in effort into identifying which images are actually appropriate; a similar choice applies in other contexts like peer-grading). A key issue in these crowdsourced judgement elicitation scenarios is therefore {\em incentivizing effort}\footnote{We thank David Evans (VP Education, Udacity) for pointing out this issue in the context of peer-grading applications--- while students might put in their best efforts on grading screening assignments to ensure they demonstrate the minimum proficiency required to be allowed to grade, how can we be sure that they will continue to work with the same proficiency when grading homeworks outside of this screening set?}--- that is, ensuring that agents {\em make} the best judgements that they possibly can (in addition, of course, to ensuring that they then  truthfully report this observation). This leads to a new kind of information elicitation problem where an agent's proficiency now {\em depends} on her effort choice, and so is {\em endogenous} and unknown to the system--- even if an agent's maximum proficiency is known, the actual proficiency with which she performs a task is an endogenous, strategic choice and therefore cannot be assumed as fixed or given.
	
%, instead  using low-effort strategies, for example, always declaring all assignments to be perfect.
A mechanism for information elicitation in this setting should make it `most beneficial', if not the only beneficial strategy, for agents to not just {\em report} their observations truthfully, but to also {\em make} the best observations they can in the first place. Also, it is even more important now to ensure that the payoffs from all agents always blindly reporting the same observation (for instance, declaring all content to be good) are strictly smaller than the payoffs from truthfully reporting what was actually observed, since declaring all tasks to be of some predecided type requires no effort and therefore incurs no cost, whereas actually putting in effort into making observations will incur a nonzero cost. Finally, unlike mechanisms designed for settings where a large audience is being polled for its opinion about a single event, a mechanism here must retain its incentive properties even when there are only a few reports per task--- this is because it can be infeasible, due to either monetary or effort constraints, to solicit reports from a large number of agents for each task. (For example, the number of tasks in peer grading scales {\em linearly} with the number of agents, limiting the number of reports available for each task since each student can only grade a few assignments; similarly, the total cost to the requester in crowdsourcing platforms such as Amazon Mechanical Turk scales linearly with the number of workers reporting on each task). How can we elicit the best possible evaluations from agents whose proficiency of evaluation depends on their strategically chosen effort, when the ground truth as well as the effort levels of agents are unobservable to the mechanism?  \\
%As discussed in \S \ref{s-relwork}, none of the existing mechanisms for information elicitation achieve all these desiderata.

\noindent {\bf Our Contributions.} We introduce a model for information elicitation with {\em endogenous proficiency}, where an agent's strategic choice of whether or not to put in effort into a task endogenously determines her proficiency (the probability of correctly evaluating the ground truth) for that task. We focus on the design of mechanisms for binary information elicitation, \ie, when the underlying ground truth is binary (corresponding to eliciting `good' or `bad' ratings). While generalizing to an arbitrary underlying type space is an immediate direction for further work, we note that a number of interesting judgement and evaluation tasks, for example identifying adult content or correctness evaluation, are indeed binary; also, even very recent literature providing improved mechanisms for information elicitation (e.g.~\cite{witkowski2012robust,witkowski_ec12}), as well as experimental work on the performance of elicitation mechanisms~\cite{prelec2007algorithm,john2012measuring}, focuses on models with binary ground truth.

Our main contribution is a simple, new, mechanism for binary information elicitation for multiple tasks when agents have endogenous proficiencies. Our mechanism has the following incentive properties.
\begin{enumerate}
\item[(i)] Exerting maximum effort followed by truthful reporting of observations is a Nash equilibrium.
\item[(ii)] This is the equilibrium with {\em maximum payoff} to all agents, {\em even} when agents have different maximum proficiencies, can use {\em mixed} strategies, and can choose a different strategy for each of their tasks.

Showing that full-effort truthtelling leads to the maximum reward amongst all equilibria (including those involving mixed strategies) requires arguing about the rewards to agents in all possible equilibria that may arise. To do this, we use a matrix representation of strategies where every strategy can be written as a convex combination of `basis' strategies, so that maximizing a function over the set of all possible strategies is equivalent to a maximization over the space of coefficients in this convex combination. This representation lets us show that the reward to an agent over all possible strategy choices (by herself and other agents), and therefore over all equilibria, is maximized when all agents use the strategy of full-effort truthful reporting.

\item[(iii)] Suppose there is some positive probability, however small, that there is some `trusted' agent for each task who will report on that task truthfully with proficiency greater than half. Then the equilibrium where all agents put in full effort and report truthfully on all their tasks is essentially the only equilibrium of our mechanism, {\em even if} the mechanism does not know the identity of the trusted agents.
\end{enumerate}

We note that our mechanism requires only minimal bounds on the priors and imposes no conditions on a diverging number of agent reports per task  to achieve its incentive properties--- to the best of our knowledge, previous mechanisms for information elicitation do not provide all these guarantees simultaneously, even when proficiency is not an endogenously determined choice (see \S\ref{s-relwork} for a discussion).

The main idea behind our mechanism $\M$ is the following. With just one task, it is difficult to distinguish between agreement arising from high-effort observations of the same ground truth, and `blind' agreement achieved by the low-effort strategy of always making the same report. We use the presence of {\em multiple} tasks and ratings to distinguish between these two scenarios and appropriately reward or penalize agents to incentivize high effort--- % from their agreement score.
%AG: From their agreement score?? Did not understand and found it confusing so commented out.
% `statistic' term corresponding to the expected frequency of agreement if agents were randomly choosing reports corresponding to their estimated means.
$\M$ rewards an agent $i$ for her report on task $j$ for agreeing with another `reference' agent $\ipj$'s report on the same task, but also {\em penalizes for blind agreement} by subtracting out a statistic term corresponding to the part of $i$ and $\ipj $'s agreement on task $j$ that is to be `expected anyway' given their reporting statistics estimated from other tasks. This statistic term is chosen so that there is no benefit to making reports that are independent of the ground truth; the incentive properties of the mechanism follow from this property that agents obtain positive rewards {\em only} when they put effort into their evaluations.

\subsection{Related Work}
\label{s-relwork}
The problem of designing incentives for crowdsourced judgement elicitation is closely related to the growing literature on information elicitation mechanisms. The key difference between this literature (discussed in greater detail below) and our work is that agents in the settings motivating past work have opinions that are {\em experientially formed} anyway--- independent, and outside of, any mechanisms to elicit opinions--- so that agents only need be incentivized to participate and truthfully report these opinions. In contrast, agents in the crowdsourcing settings we study {\em do not have such experientially formed opinions} to report--- an agent makes a judgement only because it is part of her task, expending effort to {\em form} her judgement, and therefore must be incentivized to both expend this effort and then to truthfully report her evaluation. There are also other differences in terms of the models and guarantees in previous mechanisms for information elicitation; we discuss this literature below.

The peer-prediction method, introduced by Miller, Resnick and Zeckhauser~\cite{miller2005eliciting}, is a mechanism for the information elicitation problem for general outcome spaces where truthful reporting is a Nash equilibrium, uses proper scoring rules to reward agents for reports that are predictive of other agents' reports. The main difference between our mechanism and~\cite{miller2005eliciting}, as well as other mechanisms based on the peer prediction method~\cite{jurcafaltings06,jurcafaltings07,witkowski_sc11,jurcafaltings09,witkowski_ec12}, is in the model of agent proficiency. In peer-prediction models, while agents can decide whether to incur the cost to {\em participate} (\ie, submit their opinion), an agent's proficiency--- the distribution of opinions or evaluations conditional on ground truth--- is {\em exogenously} determined (and common to all agents, and in most models, known to the center). %% REVIEWER COMM: and known to the center).
That is, an agent might not participate at all, but if she does participate she is assumed to have some known proficiency. In contrast, in our setting, agents can choose not just whether or not to participate, but also {\em endogenously} determine their proficiency conditional on participating through their effort choice. Thus, while peer prediction mechanisms do need to incentivize agents to participate (by submitting a report), they then know the proficiency of agents who do submit a report, and therefore can, and do, dispense rewards that crucially use knowledge of this proficiency. In contrast, even agents who do submit reports in our setting cannot be assumed to be using their maximum proficiency to make their evaluations, and therefore cannot be rewarded based on any assumed level of proficiency.
Additionally, truthtelling, while an equilibrium, is {\em not necessarily the maximum-reward equilibrium} in these existing peer-prediction mechanisms--- \cite{jurca2005enforce} shows that for the mechanisms in~\cite{miller2005eliciting,jurcafaltings06}, the strategies of always reporting `good' or always reporting `bad' both constitute Nash equilibria, at least one of which generates {\em higher} payoff than truthtelling. Such blind strategy equilibria can be eliminated and honest reporting made the unique Nash equilibrium by  designing the payments as in~\cite{jurcafaltings09}, but this needs agents to be restricted to pure reporting strategies, and requires full knowledge of the prior and conditional probability distributions to compute the rewards.

%The mechanism in~\cite{jurcafaltings09} restricts agents to pure reporting strategies, and is then able to eliminate such blind strategy equilibria, make honest reporting the unique pure %strategy Nash equilibrium and address a wide range of collusion scenarios.
%The mechanism in Jurca and Faltings~\cite{jurcafaltings09} extends the results of~\cite{jurcafaltings06} to design payments with knowledge of the priors and conditional distributions, that
%makes honest reporting the unique pure strategy Nash equilibria, in addition, to addressing a wide range of collusion scenarios. Their mechanism, however, is restricted only to considering pure reporting strategies,

The Bayesian Truth Serum (BTS)~\cite{prelec2004bayesian} is another mechanism for information elicitation with unobservable ground truth. BTS does not use the knowledge of a common prior to compute rewards, but rather collects two reports from each agent--- an `information' report which is the agent's own observation, as well as a `prediction' report which is the agent's prediction about the distribution of information reports from the population--- and uses these to compute rewards such that truthful reporting is the highest-reward Nash equilibrium of the BTS mechanism. In addition to the key difference of exogenous versus endogenous proficiencies discussed above, an important limitation of BTS in the crowdsourcing setting is that it requires the number of agents $n$ reporting on a task to diverge to ensure its incentive properties. This $n \rightarrow \infty$ requirement is infeasible in our setting due to the scaling of cost with number of reports as discussed in the introduction. ~\cite{witkowski2012robust} provides a robust BTS mechanism (RBTS) that works even for small populations (again in the same non-endogenous proficiency model as BTS and peer prediction mechanisms), and also ensures payments are positive, making the mechanism ex-post individually rational in contrast to BTS. However, the RBTS mechanism does not retain the property of truthtelling being the highest reward Nash equilibrium--- indeed, the `blind agreement' equilibrium via constant reports achieves the maximum possible reward in RBTS, whereas truthtelling might in fact lead to lower rewards.
%Also, both BTS and RBTS require agents to provide a prediction report for each task (\ie, a belief about how other graders will evaluate the same assignment) to achieve their incentive properties, potentially imposing a large cognitive burden on agents.
%% COMMENT: For RBTS, the maximum reward is 2, which can be achieved by setting x_i = y_i = 1 for all i. Truthtellin gives <2.

There is also work on information elicitation in conducting surveys and online polling~\cite{lambert2008truthful,jurcafb08}, both of which are not quite appropriate for our crowdsourcing setting. The mechanism in~\cite{lambert2008truthful} is weakly incentive compatible (agents are indifferent between lying and truthtelling), while \cite{jurcafb08} presents a online mechanism that is not incentive compatible in the sense that we use and potentially requires a large (constant) number of agents to converge to the true result. For other work on information elicitation, albeit in settings very different from ours, see~\cite{goelrp09,chen2012predicting,papakonstantinou2011mechanism}.

We note also that we model settings where there is indeed a notion of a ground truth, albeit unobservable, so that proficient agents who put in effort are more likely than not to correctly observe this ground truth. Peer-prediction methods, as well as the Bayesian truth serum, are designed for settings where there may be no underlying ground truth at all, and the mechanism only seeks to elicit agents' true observations (whatever they are) which means that some agents might well be in the minority even when they truthfully report their observation--- this makes the peer prediction setting `harder' along the dimension of inducing truthful reports, but easier along the dimension of not needing to incentivize agents to choose to exert effort to make high-proficiency observations.

%Lambert and Shoham~\cite{lambert2008truthful} present a mechanism for surveys that is only weakly incentive compatible, \ie, agents are indifferent
% between truthtelling and misreporting.
%which does not assume common priors,  -- THIS IS NOT IMPORTANT IN OUR CASE
% The mechanism for online polling in~\cite{jurcafb08} is also not incentive compatible in our sense, agents can choose to

We note that our problem can also be cast as a version of a principal-agent problem with a very large number of agents, although the principal cannot directly observe an agent's `output' as in standard models. While there is a vast literature in economics on the principal-agent problem too large to describe here (see, eg, ~\cite{bolton2005contract} and references therein), none of this literature, to the best of our knowledge, addresses our problem. Finally, there is also a large orthogonal body of work on the problem of {\em learning} unknown (but {\em exogenous}) agent proficiencies, as well as on the problem of optimally combining reports from agents with differing proficiencies to come up with the best aggregate evaluation in various models and settings. These problems of learning exogenous agent proficiencies and optimally aggregating agent reports are orthogonal to our problem of providing incentives to agents with endogenous, effort-dependent proficiencies to elicit the best possible evaluations from them. %\footnote{Note that in the peer-grading setting, the maximum proficiency can be estimated, to some degree at least, by having students evaluate pre-graded assignments.}.

\section{Model}
\label{s-model}
We now present a simple abstraction of the problem of designing mechanisms for crowdsourced judgement elicitation settings where agents' proficiencies are determined by strategic effort choice.

{\em Tasks.} There are $m$ tasks, or objects, $j = 1, \ldots, m$, where each task has some underlying `true quality', or type, $\bar X_j$. This true type $\bar X_j$ is unknown to the system. We assume that the types are binary-valued: $\bar X_j$ is either $H$ (or $1$, corresponding to high-quality) or $L$ (or $0$, for low quality) for all $j$; we use $1$ and $H$ (resp. $0$ and $L$) interchangeably throughout for convenience. The prior probabilities of $H$ and $L$ for all tasks are denoted by $\ph$ and $\pl$. We assume throughout that $\max(\ph,\pl) < 1$, \ie, that there is at least some uncertainty in the underlying qualities of the objects.

{\em Agents.} There are $n$ workers or agents $i = 1, \ldots, n$ who noisily evaluate, or form judgements on, the qualities of objects. We say agent $i$ performs task $j$ if $i$ evaluates object $j$. Agent $i$'s judgement on task $j$ is denoted by $\hat X_{ij} \in \{0,1\}$, where $\hat X_{ij}$ is $0$ if $i$ evaluates $j$ to be of type $L$ and $\hat X_{ij}$ is $1$ if $i$ evaluates it to be $H$. Having made an evaluation $\hat X_{ij}$, an agent can choose to report any value $X_{ij} \in \{0,1\}$ either based on, or independent of, her actual evaluation $\hat X_{ij}$.

We denote the set of tasks performed by an agent $i$ by  $J(i)$, and let $I(j)$ denote the set of agents who perform task $j$. We will assume for notational simplicity that $|J(i)| = D$ and $|I(j)| = T$ for all agents $i$ and tasks $j$. %, so that $nD = mT$.

{\em Proficiency.} An agent's {\em proficiency} at a task is the probability with which she correctly evaluates its true type or quality. We assume that an agent's proficiency is an increasing function of the {\em effort} she puts into making her evaluation. Let $e_{ij}$ denote agent $i$'s effort level for task $j$: we assume for simplicity that effort is binary-valued, $e_{ij} \in \{0,1\}$. Putting in $0$ effort has cost $c_{ij}(0) = 0$, whereas putting in full effort has cost $c_{ij}(1) \geq 0$ (we note that our results also extend to a linear model with continuous effort where $e_{ij} \in [0,1]$ and the probability $p_i(e_{ij})$ of correctly observing $\bar X_j$ as well as the cost $c_i(e_{ij})$ increase linearly with $e_{ij}$).

An agent who puts in zero effort makes evaluations with proficiency $p_{ij}(0) = 1/2$ and does no better than random guessing, \ie, $\Pr(\hat X_{ij} = \bar X_j|e_{ij} = 0) =1/2$. An agent who puts in full effort $e_{ij} = 1$ attains her {\em maximum proficiency}, $\Pr(\hat X_{ij} = \bar X_j|e_{ij} = 1) = p_{ij}(1) = p_i$. Note that this maximum proficiency $p_i$ can be different for individual agents modeling their different abilities, and need {\em not} be known to the center. We assume that the maximum proficiency $p_i \ge \frac{1}{2}$ for all $i$--- this minimum requirement on agent ability can be ensured in online crowdsourcing settings by {\em prescreening} workers on a representative set of tasks (Amazon Mechanical Turk, for instance, offers the ability to prescreen workers~\cite{paolacci2010running,harris11csdm}, whereas in peer-grading applications such as on Coursera, students are given a set of pre-graded assignments to measure their grading abilities prior to grading their peers, the results of which can be used as a prescreen.)

We note that our results also extend easily to the case where the maximum proficiency of an agent {\em depends} on whether the object is of type $H$ or $L$, \ie, the probabilities of correctly observing the ground truth when putting in full effort are different for different ground truths, $\Pr(\hat X_{ij} = \bar X_j| \bar X_j = H) \neq \Pr(\hat X_{ij} = \bar X_j| \bar X_j = L)$ (of course, different agents can continue to have different maximum proficiences).

{\em Strategies.} Agents strategically choose {\em both} their effort levels and reports on each task to maximize their total utility, which is the difference between the reward received for their reports and the cost incurred in making evaluations. Formally, an agent $i$'s strategy is a vector of $D$ tuples $[(e_{ij}, f_{ij})]$, specifying her effort level $e_{ij}$ as well as the function $f_{ij}$ she uses to {\em map} her actual evaluation $\hat X_{ij}$ into her report $X_{ij}$ for each of her tasks. Note that since an agent's proficiency on a task $p_{ij}$ is a function of her strategically chosen effort $e_{ij}$, the proficiency of agent $i$ for task $j$ is {\em endogenous} in our model.

For a single task, we use the notation $(1, \truth)$ to denote the choice of full effort $e_{ij} = 1$ and truthfully reporting one's evaluation (\ie, $f_{ij}$ is the identity function $X_{ij} = \hat X_{ij}$), $(1, \invert)$ to denote full effort followed by inverting one's evaluation, and $(0, r)$ to denote the choice of exerting no effort ($e_{ij} = 0$) and simply reporting the outcome of a random coin toss with probability $r$ of returning $H$. We use $[(1,\truth)]$ to denote the strategy of using full effort and truthtelling on all of an agent's tasks, and similarly $[(1, X^c)]$ and $[(0,r)]$ for the other strategies.

{\em Mechanisms.} A mechanism in this setting takes as input the set of all received reports $X_{ij}$ and computes a reward for each agent based on her reports, as well as possibly the reports of other agents. Note that the mechanism has no access\footnote{Crowdsourcing is used typically precisely in scenarios where the number of tasks is too large for the principal (or a set of trusted agents chosen by the principal) to carry out herself, so it is at best feasible to verify the ground truth for a tiny fraction of all tasks, which fraction turns out to be inadequate (a formal statement is omitted here) to incentivize effort using knowledge of the $\bar X_j$.} to the underlying true qualities $\bar X_j$ for any task, and so cannot use the $\bar X_j$ to determine agents' rewards. A set of effort levels and reporting functions $[(e_{ij},f_{ij})]$ is a full-information Nash equilibrium of a mechanism if no agent $i$ can strictly improve her expected utility by choosing either a different effort level $\hat e_{ij}$, or a different function $\hat f_{ij}$ to map her evaluation $\hat X_{ij}$ into her report $X_{ij}$. Here, the expectation is over the randomness in agents' noisy evaluation of the underlying ground truth, as well as any randomness in the mechanism.

We will be interested in designing mechanisms for which it is (i) an equilibrium for all agents to put in full effort and report their evaluations truthfully on all tasks, \ie, use strategies $[(1, \truth)]$, and (ii) for which $[(1, \truth)]$ is the maximum utility (if not unique) equilibrium. We emphasize here that we do not address the problem of how to optimally aggregate the $T$ reports $X_{ij}$ for task $j$ into a final estimate of $\bar X_j$: this is an orthogonal problem requiring application-specific modeling; our only goal is to elicit the best possible judgements to aggregate, by ensuring that agents find it most profitable to put in maximum effort into their evaluations and then report these evaluations truthfully.

\section{Mechanism}
\label{s-mech}
The main idea behind our mechanism $\M$ is following. Recall that a mechanism does not have access to the true qualities $\bar X_j$, and therefore must compute rewards for agents that do not rely on directly observing $\bar X_{j}$. Since the only source of information about $\bar X_{j}$ comes from the reports $X_{ij}$, a natural solution is to reward based on some form of agreement between different agents reporting on $j$, similar to the peer-prediction setting~\cite{miller2005eliciting}. However, an easy way for agents to achieve perfect agreement with no effort is to always report $H$ (or $L$). With just one task, it is difficult for a mechanism to distinguish between the scenario where agents achieve agreement by making accurate, high-effort, evaluations of the same ground truth, and the low-effort scenario where agents achieve agreement by always reporting $H$, especially if $\ph$ is high. However, in our setting, we have the benefit of {\em multiple} tasks and ratings, which could potentially be used to distinguish between these two strategies and appropriately reward agents to incentivize high effort.

$\M$ uses the presence of multiple ratings to subtract out a {\em statistic} term $B_{ij}$ from the agreement score, chosen so that there is no benefit to making reports that are independent of $\bar X_j$--- % (in fact, the reward for random reporting is precisely $0$).
roughly speaking, $\M$ rewards an agent $i$ for her report on task $j$ for agreeing with another `reference' agent $\ipj$'s report on the same task, but {\em only beyond} what would be expected if $i$ and $\ipj$ were randomly tossing coins with their respective empirical frequencies of heads.

Let $d$ denote the number of other reports made by $i$ and $\ipj$ that are used in the computation of this statistic term $B_{ij}$ based on the observed frequency of heads for each pair $(i,j)$. We use $\M_d$ to denote the version of $\M$ which uses $d$ other reports from each of $i$ and $\ipj$ to compute $B_{ij}$. To completely specify $\M_d$, we also need to specify a reference rater $\ipj $ as well as this set of $d$ (non-overlapping) tasks performed by $i$ and $\ipj$, for which we use the following notation. (We require these $d$ other tasks to be non-overlapping so that the reports for these tasks $X_{ik}$ and $X_{\ipj l}$ are independent\footnote{We assume that co-raters' identities are kept unknown to agents, so there is no collusion between $i$ and $\ipj$.}, which is necessary to achieve the incentive properties of $\M_d$.)

\begin{definition}[$S_{ij}, S_{\ipj j}$] \label{d-Sij} Consider agent $i$ and task $j \in J(i)$, and a reference rater $\ipj$. Given a value of $d$ ($1 \leq d \leq D-1$), let  $S_{ij}$ and $S_{\ipj j}$ be sets of $d$ non-overlapping tasks other than task $j$ performed by $i$ and $\ipj$ respectively, \ie,
\begin{align*}
& S_{ij} \subseteq J(i)\setminus j, \quad S_{\ipj j} \subseteq J(\ipj)\setminus j, \quad S_{ij}\cap S_{\ipj j} = \emptyset, \quad |S_{ij}|= |S_{\ipj j}| = d.
\end{align*}
%$k \neq j$ performed by $i$ such that $i$'s reports on $k$, $X_{ik}$, are used to compute the statistic term of the reward $R_{ij}$. Similarly, let $S_{\ipj j}$ be a set of $d$ tasks $k' \neq j$, $k' \notin S_{ij}$, such that $\ipj $'s reports on $k'$, $X_{\ipj k'}$, are used  to compute the statistic term of $i$'s reward for $j$. Note that there is no common task between these two sets $S_{ij}, S_{\ipj j}$, \ie, $S_{ij}\cap S_{\ipj j} = \emptyset$, for all $i,j, \ipj$.
\end{definition}

%AG: commented this out and moved to right before main theorem,since it is never used before that and spoils the flow here.
%A similar useful definition is the following.
%\begin{definition}[$T_{ij}$, $d_{ij}$]\label{d-Tij} Let $T_{ij}$ be the set of all tasks  $j' \neq j$ such that $j \in S_{ij'}$, \ie, $T_{ij}$ is the set of tasks $j'$ for which $i$'s report on task $j$ is used to compute the statistic term of $i$'s reward $R_{ij'}$ for task $j'$. We use $d_{ij} = |T_{ij}|$ to denote the number of such tasks $j'$.
%\end{definition}

A mechanism $\M_d$ is completely specified by reference raters $\ipj$ and the sets $S_{ij}$ and $S_{\ipj j}$, and rewards agents as defined below. Note that $\M_d$ only uses agents' {\em reports} $X_{ij}$ to compute rewards and not their maximum {\em proficiencies} $p_i$, which therefore need not be known to the system.

\begin{definition}[Mechanism $\M_d$] \label{d-mech} $\M_d$ computes an agent $i$'s reward for her report $X_{ij} \in \{0,1\}$ on task $j$, $R_{ij}$,  by comparing against a `reference rater' $\ipj$'s report $X_{\ipj j}$ for $j$, as follows:
\begin{align}
R_{ij} &= A_{ij} - B_{ij}, \quad \mbox{where}\\
A_{ij}&= X_{ij}X_{\ipj j} + (1-X_{ij})(1-X_{\ipj j}), \quad \mbox{and}\nonumber \\
B_{ij} &= (\frac{\sum_{k\in S_{ij}} X_{ik}}{d})(\frac{\sum_{l\in S_{\ipj j}} X_{\ipj l}}{d}) + (1-\frac{\sum_{k\in S_{ij}} X_{ik}}{d})(1-\frac{\sum_{l\in S_{\ipj j}} X_{\ipj l}}{d} ),
\end{align}
where the sets $S_{ij}$ and $S_{\ipj j}$ in $B_{ij}$ are as in Definition \ref{d-Sij}. The final reward to an agent $i$ is $\beta R_i$, where $R_i = \sum_{j \in J(i)} R_{ij}$ and $\beta$ is simply a non-negative scaling parameter that is chosen based on agents' costs of effort.
\end{definition}

The first term, $A_{ij}$, in $R_{ij}$ is an `agreement' reward, and is $1$ when $i$ and $\ipj $ both agree on their report, \ie, when $X_{ij} = X_{\ipj j} = 1$  or when $X_{ij} = X_{\ipj j} = 0$. The second term $B_{ij}$ is the `statistic' term which, roughly speaking, deducts from the agreement reward whatever part of $i$ and $\ipj $'s agreement on task $j$ is to be `expected anyway' given their reporting statistics, \ie, the relative frequencies with which they report $H$ and $L$. This deduction is what gives $\M$ its nice incentive properties--- while $\M$ rewards agents for agreement via $A_{ij}$,
%(which is necessary since this is the only source of information about agents' effort),
$\M$ also {\em penalizes for blind agreement} that agents achieve without effort, by subtracting out the $B_{ij}$ term corresponding to the expected frequency of agreement if $i$ and $\ipj $ were randomly choosing reports corresponding to their estimated means.

For example, suppose all agents were to always report $H$. Then $A_{ij}$ is always $1$, but $B_{ij}=1$ as well so that the net reward is $0$; similarly if agents chose their reports according to a random cointoss, even one with the `correct' bias $\ph$, the value of $A_{ij}$ is exactly equal to $B_{ij}$ since there is no correlation between the reports for a particular task, again leading to a reward of $0$. The reward function $R_{ij}$ is designed so that it {\em only} rewards agents when they put in effort into their evaluations, which leads to the desirable incentive properties of $\M_d$. (We note that there are other natural statistics which might incentivize agents away from low-effort reports--- \eg, rewarding reports which collectively have an empirical mean close to $\ph$, or for variance. However, it turns out that appropriately balancing the agreement term (which is necessary to ensure agents cannot simply report according to a cointoss with bias $\ph$) with a term penalizing blind agreement to simultaneously ensure that $[(1, \truth)]$ is an equilibrium {\em and} the most desirable equilibrium is hard to accomplish.)

There are two natural choices for the parameter $d$, \ie, how many reports of $i$ and $\ipj $ to include for estimating the statistic term that we subtract from the agreement score in $R_{ij}$\footnote{Understanding the effect of the parameter $d$ in our mechanisms, which appears irrelevant to the mechanism's behavior when agents are risk-neutral, is an interesting open question.}.
%\begin{itemize} \item
(i) In $\M_{D-1}$, we set $d = D-1$ and include all reports of agents $i$ and $\ipj$, except those on their common task $j$. Here, the non-overlap requirement for sets $S_{ij}$ and $S_{\ipj j}$ says that an agent $i$ and her reference rater $\ipj$ for task $j$ have {\em only} that task $j$ in common.
%(such an allocations of tasks to agents is easily achieved as described in the full version of the paper).
%AG: Sorry, reworded your wording but postponed to the full version made it sound like we postponed figuring it out which is not a good message and also false.
%AG:
%Here, the independence requirement of $\M_{D-1}$ says that the assignment of tasks to agents, and the choice of reference rater for an agent must be carried out in such a way that an agent $i$ and her reference rater $\ipj$ for task $j$ have {\em only} that task $j$ in common, for all agents $i$ and all tasks $j$.
%\item
(ii) In $\M_1$, we set $d = 1$, \ie, subtract away the correlation between the report of $i$ and $\ipj$ on exactly one other non-overlapping task.
In $\M_1$, the non-overlap condition only requires that for each agent-task pair, there is a reference agent $\ipj$ available who has rated one other task that is different from the remaining tasks rated by $i$, a condition that is much easier to satisfy than that in $\M_{D-1}$.
%Satisfying the non-overlap requirement is much easier when $d = 1$, since it is easier to find a reference rater with {\em one} non-overlapping task than a rater with $D-1$ non-overlapping tasks.
In \S~\ref{s-eq}, we will see that $\M_1$ will require that the choices of $(j,j')$, where $\{j'\} = S_{ij}$ is the task used in the statistic term of $i$'s reward for task $j$, are such that each task $j'$ performed by $i$ is used exactly once to determine $R_{ij}$ for $j \neq j'$. Note that this is always feasible, for instance by using task $j+1$ in the statistic term for task $j$ for $j = 1, \ldots, D-1$ and task $1$ for task $D$.

\section{Analyzing $\M$}
\label{s-eq}
In this section, we analyze equilibrium behavior in $\M_d$. We begin with some notation and preliminaries.

\subsection{Preliminaries}
Recall that proficiency is the probability of correctly evaluating the true quality. We use $\sph$ (respectively $\spl$) to denote the probability that an agent observes $H$ (respectively $L$) when making evaluations with proficiency $p$, \ie,
%\[
the probability that $\hat X_{ij}=H$ is $\sph = p\ph + (1 - p)\pl$.
%\]
Similarly, $q[H], q[L]$ and $p_i[H], p_i[L]$ correspond to the probabilities of seeing $H$ and $L$ when making evaluations with proficiencies $q$ and $p_i$ respectively. \\

\noindent{\em Matrix representation of strategies.} We will frequently need to consider the space of all possible strategies an agent may use in the equilibrium analysis of $\M_d$.
%Recall that an agent's strategy for a particular task consists of two components--- an effort level $e_{ij}$, and a function mapping her  evaluation $\hat X_{ij}$ (the correctness of which depends on her effort choice) into her report $X_{ij}$.
While the choice of effort level $e_{ij}$ in an agent's strategy $[(e_{ij}, f_{ij})]$ is easily described--- there are only two possible effort levels $1$ and $0$--- the space of functions $f_{ij}$ through which an agent can map her evaluation $\hat X_{ij}$ into her report $X_{ij}$ is much larger. For instance, an agent could choose $f_{ij}$ corresponding to making an evaluation, performing a Bayesian update of her prior on $\bar X_j$, and choosing the report with the higher posterior probability.
%Proving that a set of strategies constitutes an equilibrium requires demonstrating that no such strategy constitutes a beneficial deviation.
We now discuss a way to represent strategies that will allow us to easily describe the set of all reporting functions $f_{ij}$.

An agent $i$'s evaluation $\hat X_{ij}$ can also be written as a two-dimensional vector $o^{ij} \in {\mathbb R}^2$, where $o^{ij} = \begin{bmatrix} 1 & 0 \end{bmatrix}^T$ if $i$ observes a $H$, and $o^{ij} = \begin{bmatrix} 0 & 1 \end{bmatrix}^T$ if $i$ observes a $L$, where ${\bf a}^T$ denotes the transpose of ${\bf a}$. For the purpose of analyzing $\M_d$, any choice of reporting function $f_{ij}$ can then be described via a $2 \times 2$ matrix
\[
M^{ij} = \begin{bmatrix}
x & 1-y\\
1-x & y
\end{bmatrix},
\]
where $x$ is the probability with which $i$ chooses to report $H$ after observing $H$, \ie, $x = \Pr(X_{ij} = H | \hat X_{ij} = H)$, and similarly $y = \Pr(X_{ij} = L | \hat X_{ij} = L)$. Observe that the choice of effort $e_{ij}$ affects {\em only} $o^{ij}$ and its `correctness', or correlation with the (vector representing the) actual quality $\bar X_j$, and the choice of reporting function $f_{ij}$ {\em only} affects $M^{ij}$.

Any reporting matrix $M^{ij}$ of the form above can be written as a convex combination of four matrices--- one for each of the $f_{ij}$ corresponding to (i) truthful reporting ($X_{ij} = \hat X_{ij}$) (ii) inverting ($X_{ij} = \hat X_{ij}^c$), and (iii, iv) always reporting $H$ or $L$ independent of one's evaluation ($X_{ij} = H$ and $X_{ij} = L$ respectively):
\[
M_{X} = \begin{bmatrix}
1 & 0\\
0 & 1
\end{bmatrix},
M_{X^c} = \begin{bmatrix}
0 & 1\\
1 & 0
\end{bmatrix},
M_{H} = \begin{bmatrix}
1 & 1\\
0 & 0
\end{bmatrix},
M_{L} = \begin{bmatrix}
0 & 0\\
1 & 1
\end{bmatrix}.
\]
That is, $M^{ij} = \al_1 M_{X} + \al_2 M_{X^c} + \al_3 M_{H} + \al_4 M_{L}$,  where $\al_1 = x-\al_3$, $\al_2 = 1-y-\al_3$, and $\al_3 = x-y$ and $\al_4 = 0$ if $x \geq y$, and $\al_3 = 0$ and $\al_4 = y-x$ if $y > x$. It is easily verified that $\al_i \geq 0$, and $\sum \al_i = 1$, so that this is a convex combination.
Since all possible reporting strategies $f_{ij}$ can be described by appropriately choosing the values of $x \in [0,1]$ and $y \in [0,1]$ in $M^{ij}$, every reporting function $f_{ij}$ can be written as a convex combination of these four matrices.
%We will sometimes use $M_r =  rM_H + (1-r)M_L$ to denote randomly reporting $H$ with probability $r$ and $L$ with probability $1-r$.

The agent's final report $X_{ij}$ is then described by the vector $M^{ij}o^{ij} \in {\mathbb R}^2$, where the first entry is the probability that $i$ reports $H$, \ie, $X_{ij} = H$, and the second entry is the probability that she reports $X_{ij}= L$. The expected reward of agent $i$ for task $j$ can therefore be written using the matrix-vector representation (where $^T$ denotes transpose and $\ones$ is the vector of all ones) as
\begin{align*}
E[R_{ij}] &= E[(M^{\ipj j}o^{\ipj j})^T M^{ij}o^{ij} + (\ones-M^{\ipj j}o^{\ipj j})^T(\ones-M^{ij}o^{ij})] \\
&-  [(M^{\ipj j}E[o^{\ipj j}])^TM^{ij}E[o^{ij}] + (\ones-M^{\ipj j}E[o^{\ipj j}])^T(1-M^{ij}E[o^{ij}])],
\end{align*}

which is linear in $M^{ij}$. So the payoff from an arbitrary reporting function $f_{ij}$ can be written as the corresponding linear combination of the payoffs from each of the `basis' functions (corresponding to $M_X, M_{X^c}, M_H$ and $M_L$) constituting $f_{ij} = M^{ij}$. We will use this to argue that it is adequate to consider deviations to each of the remaining basis reporting functions and show that they yield strictly lower reward to establish that $[(1,X)]$ is an equilibrium of $\M_d$. \\

%Which strategies are equivalent
\noindent{\em Equivalent strategies.} For the equilibrium analysis, we will use the following simple facts.
%\BIT
%\item
(i) The strategy $(0,X)$ (\ie, using zero effort but truthfully reporting one's evaluation) is equivalent to the strategy $(0,r)$ with $r = 1/2$, \ie, to the strategy of putting in no effort, and randomly reporting $H$ or $L$ independent of the evaluation $\hat X_{ij}$ with probability $1/2$ each.
%\item
(ii) The strategy $(1,r)$ is equivalent to the strategy $(0,r)$, since the report $X_{ij}$ in both cases is completely independent of the evaluation $\hat X_{ij}$ and therefore of $e_{ij}$.\\
%\EIT

%Cost doesn't really figure; enough to show strict difference in reward and scale.
\noindent{\bf Cost of effort.} While agents do incur a higher cost when using $e_{ij} = 1$ as compared to $e_{ij} = 0$, we will not need to explicitly deal with the cost in the equilibrium analysis--- if the reward from using a strategy where $e_{ij}=1$ is strictly greater than the reward from any strategy with $e_{ij}=0$, the rewards $R_{ij}$ can always be scaled appropriately using the factor $\beta$ (in Definition \ref{d-mech}) to ensure that the net utility (reward minus cost) is strictly greater as well.

We remark here that bounds on this scaling factor $\beta$ could be estimated empirically without requiring knowledge of the priors by estimating the cost of effort $c_{ij}$ from the maximum proficiencies obtained from a pre-screening (\S \ref{s-model}), by conducting a series of trials with increasing rewards and then using individual rationality to estimate the cost of effort from observed proficiencies in these trials.  \\

\noindent{\bf Individual rationality and non-negativity of payments.} The expected payments made by our mechanism to each agent are always nonnegative in the full-effort truthful reporting equilibrium, \ie, when all agents use strategies $[(1,X)]$. To ensure that the payments are also non-negative for every instance (of the tasks and reports) and not only in expectation, note that it suffices to add $1$ to the payments currently specified, since the penalty term $B_{ij}$ in the reward $R_{ij}$ is bounded above by $1$.
We also note that individual rationality can be achieved by using a value of $\beta$ large enough to ensure that the net utility $\beta R_{ij} - c(1)$ remains non-negative for all values of $\ph$--- while the expected payment $R_{ij}$ does go to zero as $\ph$ tends to $1$ (\ie, in the limit of vanishing uncertainty as the underlying ground truth is more and more likely to always be $H$ (or always be $L$)), as long as there is some bound $\epsilon > 0$ such that $\max\{\ph, \pl\} \leq 1- \eps$, a simple calculation can be used to determine a value $\beta^*(\eps)$ such that the resulting mechanism with $\beta  = \beta^*$ leads to nonnegative utilities for all agents in the full-effort truth-telling Nash equilibrium of $\M$.

\subsection{Equilibrium analysis}
We now analyze the equilibria of $\M_d$. Throughout, we index the tasks $J(i)$ corresponding to agent $i$ by $j \in \{1,\ldots,D\}$.

First, to illustrate the idea behind the mechanism, we prove the simpler result that $[(1, \truth)]$ is an equilibrium of $\M$ when agents all have equal proficiency $p_i = p$, and are restricted to choosing one common strategy for all their tasks.
\begin{proposition}
\label{p-easy}
Suppose all agents have the same maximum proficiency $p$, and are restricted to choosing the same strategy for each of their tasks. Then, all agents choosing $[(1, \truth)]$ is an equilibrium of $\M_d$ for all $d$, if $p \neq 1/2$.
\end{proposition}
\begin{proof}
Consider an agent $i$, and suppose all other agents use the strategy $(1, \truth)$ on all their tasks. As discussed in the preliminaries,
an agent's reward is linear in her reports fixing the strategies of other agents, so
it will be enough to show that there is no beneficial deviation to $(1,\invert)$, or $(0, r)$ for any $r \in [0,1]$ to establish an equilibrium.
 (as noted earlier, the choice of effort level is irrelevant when reporting $X_{ij}$ according to a the outcome of a random coin toss independent of the observed value of $X_{ij}$). The reward to agent $i$ when she uses strategy $[(1,\truth)]$ is
\begin{align*}
E[R_i((1,\truth))] &= \sum_{j=1}^D E\left[X_{ij}X_{\ipj j} + (1-X_{ij})(1-X_{\ipj j})\right]\\
&- E\left[\frac{\sum_{k \in S_{ij}} X_{ik}}{d}\frac{\sum_{l \in S_{\ipj j}}  X_{\ipj l}}{d} + (1-\frac{\sum_{k \in S_{ij}}  X_{ik}}{d})(1-\frac{\sum_{l \in S_{\ipj j}}  X_{\ipj l}}{d} )]\right]\\
&= D\left[p^2 + (1-p)^2 - (\sph^2 + (1-\sph)^2)\right]\\
&= D(p- p(H))(p- p(L))\\
&= D(2p-1)^2 \ph\pl,
\end{align*}
which is strictly positive if $p \neq 1/2$ and $\min(\ph,\pl) > 0$ (as assumed throughout), where we use $p - \sph = (2p-1)\pl$ and $p - \spl = (2p-1)\ph$.
The expected reward from deviating to $(1,\invert)$, when other agents are using $(1,\truth)$ is
\begin{align*}
E[R_i((1,\invert))] &= D\left(2p(1-p)- 2\sph (1-\sph)\right) = -D(p-p(H))(p- p(L)).
\end{align*}
Therefore, the expected reward from deviating to $(1,\invert)$ is negative and strictly smaller than the reward from $(1,\truth)$ if $p \neq 1/2$.
Finally, suppose agent $i$ deviates to playing $(0,r)$, \ie, reporting the outcome of a random coin toss with bias $r$ as her evaluation of $X_{ij}$. Her expected reward from using this strategy when other agents play according to $(1,\truth)$ is
\begin{align*}
E[R_i((0,r))] &= D\left(r\sph + (1-r)\sph - (r\sph + (1-r)(1-\sph))\right) = 0.
\end{align*}
(In fact, if either agent reports ratings on her tasks by tossing a random coin with any probability $r \in [0,1]$, independent of the underlying true realization of $X_{ij}$, the expected reward to agent $i$ is $0$.)
Therefore, if $p \neq 1/2$, deviating from $(1,\truth)$ leads to a strict decrease in reward to agent $i$. Hence, the rewards $R_{ij}$ can always be scaled appropriately to ensure that $[(1,\truth)]$ is an equilibrium of $\M$ for any values of the costs $c_i$.
\end{proof}

We will now move on to proving our main equilibrium result for $\M_d$, where agents can have different maximum proficiencies, as well as possibly use different strategies for different tasks. We begin with a technical lemma and a definition.

\begin{lemma}
\label{l-fpq}
Let $f_{\alpha}(p,q) = pq+(1-p)(1-q) - \alpha(\sph\sqh + (1-\sph)(1-\sqh))$.
If $\alpha \leq 1$,
%\begin{enumerate}\item
(i) $f_{\alpha}(p,q)$ is strictly increasing in $p$ if $q > 1/2$, and strictly increasing in $q$ if $p > 1/2$.
%\item
(ii) $f_{\alpha}(p,q)$ is nonnegative if $p, q \geq 1/2$, and positive if $p,q > 1/2$.
%\end{enumerate}
(iii) Denote $f(p, q) \triangleq f_1(p, q)$. Then, $f(p, q) = f(q, p) = f(1-p, 1-q)$. Also $f(p, 1-q) = f(1-p, q) = -f(p, q)$.
\end{lemma}

\begin{proof}
Recall that $\sph = p\ph + (1-p)\pl$, and similarly for $\sqh$.
{\small
\begin{align*}
f_{\alpha}(p,q)  &= p(2q - 1) + (1-q) -\alpha(p\ph+(1-p)\pl)(2\sqh - 1) - (1-\sqh) \\
& = p\left[(2q-1)- \alpha(\ph-\pl)(2\sqh - 1)\right] +  K_{-p}\\
&= p(2q-1)(1-\alpha(\ph-\pl)^2) + K_{-p},
\end{align*} }
where $K_{-p}$ is a term that does not depend on $p$, and we use $2\sqh - 1 = (2q-1)(\ph-\pl)$ in the last step.

Note that $\ph - \pl < \ph < 1$ if $\max(\ph,\pl) < 1$, so that $1-\alpha(\ph-\pl)^2 > 0$ if $\alpha \leq 1$. Therefore, $f_{\alpha}(p,q)$ is linear in $p$ with   strictly positive coefficient when $q > 1/2$ and $\alpha \leq 1$. An identical argument can be used for $q$ since $f_{\alpha}(p,q)$ can be written as a linear function of $q$ exactly as for $p$:
\begin{align*}
f_{\alpha}(p,q) = q(2p-1)(1-\alpha(\ph-\pl)^2) + K_{-q}.  	
\end{align*}
This proves the first claim.

For nonnegativity of $f_{\alpha}(p,q)$ on $p \in [1/2,1]$, we simply argue that $f_{1}(p,q)$ is increasing in $q$ when $p \in [1/2,1]$, and $0$ at $q = 1/2$. So for any $q > 1/2$, $f_1(p, q) \geq 0$ for any $p \in [1/2,1]$. But $f_{\alpha}(p,q)$ is decreasing in $\alpha$, so $f_{\alpha}(p,q)$ is nonnegative for any $\alpha \leq 1$ as well.

The final claims about $f(p, q)$ and $f(1-p, q)$ can be verified just by substituting the definitions of $\sph$ and $\sqh$ and from symmetry in $p$ and $q$.
\end{proof}

\begin{definition}[$T_{ij}$, $d_{ij}$]\label{d-Tij} Let $T_{ij}$ be the set of all tasks  $j' \neq j$ such that $j \in S_{ij'}$, \ie, $T_{ij}$ is the set of tasks $j'$ for which $i$'s report on task $j$ is used to compute the statistic term of $i$'s reward $R_{ij'}$ for task $j'$. We use $d_{ij} = |T_{ij}|$ to denote the number of such tasks $j'$.
\end{definition}
%\begin{proof}
%Recall that $\sph = p\ph + (1-p)\pl$, and similarly for $\sqh$.
%\begin{align*}
%f_{\alpha}(p,q) &= p(2q - 1) + (1-q) -\alpha(p\ph+(1-p)\pl)(2\sqh - 1) - (1-\sqh) \\
%& = p\left[(2q-1)- \alpha(\ph-\pl)(2\sqh - 1)\right] +  K_{-p}\\
%&= p(2q-1)(1-\alpha(\ph-\pl)^2) + K_{-p},
%\end{align*}
%where $K_{-p}$ is a term that does not depend on $p$, and we use $2\sqh - 1 = (2q-1)(\ph-\pl)$ in the last step.
%
%Note that $\ph - \pl < \ph < 1$ if $\max(\ph,\pl) < 1$, so that $1-\alpha(\ph-\pl)^2 > 0$ if $\alpha \leq 1$. Therefore, $f_{\alpha}(p,q)$ is linear in $p$ with   strictly positive coefficient when $q > 1/2$ and $\alpha \leq 1$. An identical argument can be used for $q$ since $f_{\alpha}(p,q)$ can be written as a linear function of $q$ exactly as for $p$:
%\begin{align*}
%f_{\alpha}(p,q) = q(2p-1)(1-\alpha(\ph-\pl)^2) + K_{-q}.  	
%\end{align*}
%This proves the first claim.
%
%For nonnegativity of $f_{\alpha}(p,q)$ on $p \in [1/2,1]$, we simply argue that $f(p,q)$ is increasing in $q$ when $p \in [1/2,1]$, and $0$ at $q = 1/2$ for $\alpha = 1$. So for any $q > 1/2$, $f(p,q) \geq 0$ for any $p \in [1/2,1]$ when $\alpha = 1$. But $f_{\alpha}(p,q)$ is decreasing in $\alpha$, so $f_{\alpha}(p,q)$ is nonnegative for any $\alpha \leq 1$ as well.
%
%The final claims about $f(p, q)$ and $f(1-p, q)$ can be verified just by substituting the definitions of $\sph$ and $\sqh$ and from symmetry in $p$ and $q$.
%\end{proof}

Our main equilibrium result states that under a mild set of conditions on the choice of reference raters $\ipj$ and sets $T_{ij}$, exerting full effort and reporting truthfully on all tasks is an equilibrium of $\M_d$--- {\em even when} agents have different maximum proficiencies and can choose a different strategy for each task (for instance, an agent could choose to shirk effort on some tasks and put in effort on the others). The main idea behind this result can be understood from the proof of Proposition {p-easy} above, where all agents had the same maximum proficiency $p_i = p$ and were restricted to using the same strategy for each task. There, the expected payoff from using $[(1,\truth)]$ is exactly $f(p,p)$ where $f$ is as defined in Lemma \ref{l-fpq}, while the payoff from playing $[(0,r)]$ is $0$ (independent of other agents' strategies); the payoff from deviating to $[(1,\invert)]$ when other agents play $[(1,\truth)]$ is $-f(p,p)$. Since $f(p,p) > 0$ for $p > 1/2$ and increases with $p$, it is a best response for every agent to attain maximum proficiency and truthfully report her evaluation.

Extending the argument when agents can have both different maximum proficiencies $p_i$ and use different strategy choices for each task requires more care, and are what necessitate the conditions on the task assignment in Theorem \ref{t-eqbm} below. We note that these conditions on $\M_d$ arise because of the generalization to {\em both} differing  abilities $p_i$ {\em and} being allowed to choose a different strategy for each task--- if either generalization is waived, \ie, if agents can choose different strategies per task but all have equal ability ($p_i = p$), {\em or} agents can have different abilities $p_i$ but are restricted to choosing the same strategy for all their tasks, $[(1,X)]$ can be shown to be an equilibrium of $\M_d$ even without imposing these conditions.

\begin{theorem}
\label{t-eqbm}
Suppose $p_i > 1/2$ for all $i$, and for each agent $i$, for each task $j \in J(i)$, (i) $d_{ij} = d$, and (ii) $E[p_{\ipj}] = E_{j_l \in T_{ij}}[p_{\ipjl}] \triangleq \pib$, where the expectation is over the randomness in the assignment of reference raters to tasks and the sets $T_{ij}$.
Then, $[(1, \truth)]$ is an equilibrium of $\M_d$.
\end{theorem}

The first condition in Theorem~\ref{t-eqbm}, $d_{ij}=d$, says that each task $j$ performed by an agent $i$ must contribute to computing the reward via the statistic term for exactly $d$ other tasks in $J(i)$, where $d$ is the number of reports used to compute the `empirical frequency' of $H$ reports by $i$ in the statistic term. The second condition $E[p_{\ipj}] = E_{j_l \in T_{ij}}[p_{r_{j_l}(i)}]$ says that an agent $i$ should expect the average proficiency of her reference rater $\ipj$ to be equal for all the tasks that she performs, \ie, agent $i$ should not be able to identify any particular task where her reference raters are, on average, worse than the reference raters for her other tasks (intuitively, this can lead to agent $i$ shirking effort on this task being a profitable deviation). The first condition holds for each of the two specific mechanisms $\M_1$ and $\M_{D-1}$, and the second condition can be satisfied, essentially, by a randomization of the agents before assignment, as described in  \S\ref{s-graph}. We now prove the result.
\begin{proof}
Consider agent $i$, and suppose all other agents use strategy $[(1, \truth)]$, \ie, put in full effort with truthtelling on all their tasks. It will be enough to consider pure strategy deviations, and show that there is no beneficial deviation to $(1,\invert)$, or $(0, r)$ for any $r \in [0,1]$ on any single task or subset of tasks.

First, consider a particular assignment of reference raters $\ipj$ and the sets $S_{ij}$ (and therefore $T_{ij}$). The total expected reward to agent $i$ from all her $D$ tasks in this assignment, when other agents all play according to $[(1, \truth)]$ is
%\iffalse
{\small
\begin{align*}
& E[R_i] = \sum_{j=1}^D E[X_{ij}X_{\ipj j} +  (1-X_{ij})(1-X_{\ipj j})] - \left[\frac{\sum_{k \in S_{ij}} E[X_{ik}]}{d}\sphipj + (1-\frac{\sum_{k \in S_{ij}} E[X_{ik}]}{d})(1-\sphipj)\right]\\
&= \sum_{j=1}^D E[X_{ij}X_{\ipj j} +  (1-X_{ij})(1-X_{\ipj j})] - \left[\frac{\sum_{k \in S_{ij}} E[X_{ik}]}{d}(2\sphipj-1)  + (1-\sphipj)\right]\\
&= \sum_{j=1}^D \Big[ E\left[X_{ij}X_{\ipj j} + (1-X_{ij})(1-X_{\ipj j})\right] - \sum_{j_{l}\in T_{ij}} \left(\frac{E[X_{ij}] }{d}(2p_{r_ {j_l}(i)}[H]-1)\right)  - (1-\sphipj) \Big],
\end{align*} }
where the expectation is over any randomness in the strategy of $i$ as well as randomness in $i$ and $\ipj$'s evaluations for each task $j$, and we rearrange to collect $X_{ij}$ terms in the last step.

Now, agent $i$ can receive different reference raters and task sets $S_{ij}$ in different assignments.
% (the randomness in $S_{\ipj j}$ does not matter to this calculation since we are assuming that all other agents play the same strategy $(1, \truth)$ on all their tasks).
So to compute her expected reward, agent $i$ will also take an expectation over the randomness in the assignment of reference raters to tasks and the sets $S_{ij}$, which appear in the summation above via $T_{ij}$.
%{\it Condition on assignment.} Suppose $\M_d$ is such that the expected proficiency of the reference raters for $i$ for task $j$ is equal to the expected proficiency of the raters in the set $T_{ij}$,
%\[
%E[p_{\ipj }] = E_{j_l \in T_{ij}}[p_{\phi_{i j_l} }],
%\]
%where the first expectation is over the random assignment of reference raters $\ipj $, and the second expectation is over the randomness in the choice of the set $S_{ij}$ which determines $T_{ij}$

Recall the condition that $E[p_{\ipj}] =  E_{j_l \in T_{ij}}[p_{\ipjl}] \triangleq \pib$. Using this condition and taking the expectation over the randomness in the assignments of $\ipj $ and $S_{ij}$, the expected reward of $i$ is
{\small
\begin{align*}
E[R_i]  &= \sum_{j=1}^D \Big[ E\left[X_{ij}X_{\ipj j} + (1-X_{ij})(1-X_{\ipj j})\right] - \sum_{j_{l}\in T_{ij}} \left(\frac{E[X_{ij}] }{d}(2\sphib-1)\right)  - (1-\sphib) \Big]\\
&= \sum_{j=1}^D \Big[ E\left[X_{ij}X_{\ipj j} + (1-X_{ij})(1-X_{\ipj j})\right] - \frac{d_{ij}}{d}E[X_{ij}] (2\sphib-1) - (1-\sphib) \Big],
\end{align*} }
where $\sphib = E[p_{\ipj}[H]]=E_{j_l \in T_{ij}}[p_{\ipjl}[H]]$.

The expected reward to agent $i$, when she makes evaluations with proficiency $q_j$ for task $j$ and truthfully reports these evaluations ($X_{ij} = \hat X_{ij}$), is then
{\small
\begin{align}
\label{erew-pt}
E[R_i] & = \sum_{j=1}^D \Big[q_j\pib + (1-q_j)(1-\pib) -\frac{d_{ij}}{d}\sqjh(2\sphib-1) - (1-\sphib)  \Big]\nonumber\\
&= \sum_{j=1}^D \Big[q_j\pib + (1-q_j)(1-\pib) - \frac{d_{ij}}{d}\left(\sqjh\sphib + (1-\sqh)(1-\sphib)\right) - (1-\frac{d_{ij}}{d})(1-\sphib)\Big]\nonumber\\
&= \sum_{j=1}^D \left[f_{\frac{d_{ij}}{d}}(q_j,\pib) + (\frac{d_{ij}}{d} - 1)(1-\sphib)\right].
\end{align} }
where the expectation is taken over randomness in all agents' evaluations, as well as over randomness in the choices of $\ipj $ and $S_{ij}$.

We can now show that choosing full effort and truthtelling on all tasks is a best response when all other agents use $[(1, \truth)]$ if $d_{ij} = d$. First, by Lemma \ref{l-fpq}, $f_{\frac{d_{ij}}{d}}(q_j,\pib)$ is increasing in $q_j$ provided $\frac{d_{ij}}{d} \leq 1$, so agent $i$ should choose full effort to achieve her maximum proficiency $p_i$ on all tasks.
Next, note that in terms of the expected reward, using proficiency $q_j$ and reporting $\invert$ is equivalent to using proficiency $1 - q_j$ and reporting $X$.
%%%Next, note that in terms of the expected reward, a strategy $(q_j, \invert)$ is equivalent to the strategy $(1-q_j, \truth)$.
So again by Lemma \ref{l-fpq}, deviating to $\invert$, \ie, $(1-q_j)$, on any task is strictly dominated by $\truth$ for $q_j > 1/2$ and $\pib > 1/2$.

Finally, if agent $i$ chooses $f_{ij}$ as the function which reports the outcome of a random cointoss with probability $r$ of $H$ for any task $j$, the component of $E[R_i]$ contributed by the term corresponding to $X_{ij}$ becomes
{\small
\begin{align*}
& E\left[X_{ij}X_{\ipj j} + (1-X_{ij})(1-X_{\ipj j})\right] - \frac{d_{ij}}{d}\left(E[X_{ij}]p_{\ipj}[H]  + (1-E[X_{ij}])(1-p_{\ipj}[H])\right)\\
&= rp_{\ipj}[H] + (1-r)(1-p_{\ipj}[H]) - (rp_{\ipj}[H] + (1-r)(1-p_{\ipj}[H]))) \\
&=0,
\end{align*} }
which is strictly smaller than the reward from $f_{ij} = \truth$ in (\ref{erew-pt}) if $q_j > 1/2$ and $\frac{d_{ij}}{d} \ge 1$, since $f(q_j, \pib)$ is strictly positive when $q_j, \pib > 1/2$ by Lemma \ref{l-fpq}.

Since we need $\frac{d_{ij}}{d} \leq 1$ to ensure that $(1, \invert)$ is not a profitable deviation, and $\frac{d_{ij}}{d} \geq 1$ to ensure that $(0,r)$ is not a profitable deviation, requiring $d_{ij} = d$ simultaneously satisfies both conditions. Therefore, if $d_{ij} = d$, deviating from $(1,\truth)$ on any task $j$ leads to a strict decrease in reward to agent $i$. Since the total reward to agent $i$ can be decomposed into the sum of $D$ terms which each depend only on the report $X_{ij}$ and therefore the strategy for the single task $j$, any deviation from $[(1,X)]$ for any single task or subset of tasks strictly decreases $i$'s expected reward.

Therefore, the rewards $R_i = \sum_{j \in J(i)} R_{ij}$ can always be scaled appropriately to ensure that $[(1,\truth)]$ is an equilibrium of $\M_d$.
\end{proof}

\noindent{\bf Other equilibria.} While $[(1,X)]$ is an equilibrium, $\M_d$ can have other equilibria as well--- for instance, the strategy $[(0,r)]$, where all agents report the outcome of a random cointoss with bias $r$ on each task, is also an equilibrium of $\M_d$ for all $r \in [0,1]$, albeit with $0$ reward to each agent. In fact, as we show in the next theorem, no equilibrium, symmetric or asymmetric, in pure or mixed strategies, can yield higher reward\footnote{
We note that another equilibrium which achieves the same maximum expected reward is $[(1, X^c)]$, where all agents put in full effort to make their evaluations, but then all invert their evaluations for their reports. However, $[(1, X^c)]$ is a rather unnatural, and risky, strategy, and one that is unlikely to arise in practice. Also, as we will see later, $[(1, X^c)]$ can also lead to lower rewards when there are some agents who always report truthfully.} to agents
than $[(1, \truth)]$, as long as agents `treat tasks equally' (for example, while an agent may choose to shirk effort on one task and work on all others, each of her tasks is equally likely to be the one she shirks on). We will refer to this as tasks being `apriori equivalent', so that agents cannot distinguish between tasks prior to putting in effort on them (or equivalently, the assignment of reference raters is such that agents will not find it beneficial (in terms of expected reward) to use a different strategy for a specific task). Note that this assumption is particularly reasonable in the context of applications where agents are recruited for a collection of similar tasks as in crowdsourced abuse/adult content identification, or in peer grading where each task is an anonymous student's solution to the same problem.

\begin{theorem}
\label{t-maxeq}
Suppose $p_i > 1/2$, and tasks are apriori equivalent. Then, the equilibrium where all agents choose $[(1, \truth)]$ yields maximum reward to each agent.
\end{theorem}
\begin{proof}
Consider a particular agent $i$ and task $j$, and a single potential reference rater $\ipj$ for $(i,j)$. Recall from the preliminaries that agent $i$'s choice of $f_{ij}$ can be described via a matrix $M = \al_1 M_{X} + \al_2 M_{X^c} + \al_3 M_{H} + \al_4 M_{L}$, and that we denote $i$'s evaluation via a vector $o$, where $o = [1 \quad 0]^T$ if $i$ observes $H$ and $o = [0 \quad 1]^T$ if $i$ observes $L$. Similarly, let us describe $\ipj$'s choice of reporting function via the matrix $M'$ with corresponding coefficients $\al'_i$, and denote $\ipj$'s evaluation by $o'$.

Since tasks are apriori equivalent, each player $i$ (hence $\ipj$ too) uses strategies such that $E[X_{ij}] = E[X_{ik}]$ for all $j,k \in J(i)$. Then, we can rewrite the expected reward for agent $i$ on task $j$, when paired with reference rater $\ipj$, as
{\small
\begin{align*}
E[R_{ij}] =  2(E[X_{ij}X_{\ipj j}] - E[X_{ij}]E[X_{\ipj j}]).
\end{align*}}
Using the matrix-vector representation, substituting $M, M'$ with their representations in terms of the basis matrices and expanding, and evaluating the matrix-matrix products, we have
\[
X_{ij}X_{\ipj j}  = o'^TM'^T M o = o'^T R_M o,
\]
where
\begin{align*}
R_M &= (\al_1\al'_1  + \al_2\al'_2)I + \al_2\al'_1M_{X^c} + \al_1\al'_2M^T_{X^c} + (\al_3\al'_3  + \al_4\al'_4)\ones  + (\al_3\al'_1  + \al_4\al'_2)M_H + (\al_1\al'_3  + \al_2\al'_4)M^T_H \\
&+ (\al_4\al'_1  + \al_3\al'_2)M_L + (\al_1\al'_4  + \al_2\al'_3)M^T_L,
\end{align*}
and $I, \ones$ denote the identity matrix and the matrix of all ones in $R^{2 \times 2}$ respectively, and we use $M_{X}^TM_{X} = M_{X^c}^TM_{X^c} = I$, $M_H^T M_H = M_L^T M_L = \ones$, $M_{X^c}^T M_H = M_L$, $M_{X^c}^T M_L = M_H$, and $M_H^T M_L = \zeros$. Similarly,
\begin{align*}
E[X_{ij}]E[X_{\ipj j}] = E[o'^TM'^T]E[M o] = E[o'^T]R_M E[o],
\end{align*}
where $R_M$ is as defined above.

Now, note that $M_Ho = [o_1 + o_2 \quad 0]^T = [1 \quad 0]^T$ since $o_1 + o_2 = 1$ for any evaluation vector $o$ by definition, so that $E[o'^T M_Ho] = E[o'^T]E[M_Ho]$, since $M_Ho$ is a constant. The same is the case for each of the terms $E[o'^T\ones o], E[o'^T M^T_H o], E[o'^T M^T_L o], E[o'^T M_L o]$. Therefore, these terms cancel out when taking the difference\\ $E[X_{ij}X_{\ipj j}] - E[X_{ij}]E[X_{\ipj j}]$ (corresponding to the reward from either agent choosing to report $X_{ij}$ independent of her evaluation being $0$). Also note that $E[o^{ij}] = \begin{bmatrix} \sph & \spl \end{bmatrix}^T$ if agent $i$ makes evaluations with proficiency $p$. Suppose the agents use effort leading to proficiencies $p$ and $p'$ respectively. Then, we have
\begin{align*}
E[X_{ij}X_{\ipj j}] - E[X_{ij}]E[X_{\ipj j}] &= (\al_1\al'_1  + \al_2\al'_2)(E[o'^To]- E[o']^TE[o])+~\al_2\al'_1(E[o'^TM_{X^c}o]-E[o'^T]M_{X^c}E[o]) \\
& \quad+~\al_1\al'_2(E[o'^TM^T_{X^c}o]-E[o'^T]M^T_{X^c}E[o])\\
&= (\al_1\al'_1  + \al_2\al'_2)(E[o'_1o_1 + o_2o_2'] - E[o'_1]E[o_1] - E[o_2]E[o_2']) \\
&\quad + ~ (\al_2\al'_1 + \al_1\al'_2)(E[o'_1o_2 + o_1o_2'] - E[o'_1]E[o_2] - E[o_1]E[o_2'])\\
&=(\al_1\al'_1  + \al_2\al'_2)\Big(pp' + (1-p)(1-p') - \sph p'[H] - \\
& \quad (1-\sph)(1-p'[H])\Big) +  (\al_2\al'_1 + \al_1\al'_2)\Big(p(1-p')~ + \\
&\quad (1-p)p' - \sph(1- p'[H]) - (1-\sph)p'[H]\Big).
\end{align*}
Now, note that multiplier of $(\al_1\al'_1  + \al_2\al'_2)$ is precisely $f(p,p')$, which by Lemma \ref{l-fpq} is nonnegative if $p,p'\geq 1/2$, and strictly positive if $p, p' > 1/2$.  Also, for $p, p' \geq 1/2$, note that $\sph \leq p$ and $p'[H] \leq p'$. Now, the function $g(x,y) = x(1 - y) + y(1 - x)$ is decreasing in both $x$ and $y$ for $x, y \in [\frac{1}{2},1]$ (taking derivatives), so the multiplier of $(\al_2\al'_1 + \al_1\al'_2)$ is non-positive, and negative if $p,p' > 1/2$.

So the maximum value that $E[X_{ij}X_{\ipj j}] - E[X_{ij}]E[X_{\ipj j}]$ can take for nonnegative coefficients with $\sum \al_i = \sum \al'_i = 1$, is $f(p,p')$, which is obtained by setting $\al_3 = \al_4 = 0$, $\al'_3 = \al'_4 = 0$ (\ie, with no weight on random independent reporting), and $a_1 = \al'_1 = 1$, $\al_2 = \al'_2 = 0$ (or viceversa): this is because the maximum value of term $(\al_1\al'_1  + \al_2\al'_2)$ when $\al_2 = 1-\al_1$ and $\al'_2 = 1-\al'_1$ is $1$ and is achieved with these values, which also minimize the value of the term $(\al_2\al'_1 + \al_1\al'_2)$ with the non-positive multiplier, since $(\al_2\al'_1 + \al_1\al'_2) \geq 0$ and is equal to $0$ for these values of $\al_i, \al'_i$. Also, since $f(p,p')$ increases with increasing $p$ and $p'$, it is maximized when agents put in full effort and achieve their maximum proficiencies $p_i, p_{\ipj}$.

Therefore the expected reward for the single component of $E[R_{ij}]$ coming from a specific reference rater achieves its upper bound when both agents use $[(1,X)]$. The same argument applies for each reference rater, and therefore to the expected reward $E[R_{ij}]$, and establishes the claim.
\end{proof}

We next investigate what kinds of Nash equilibria might exist where agents use low effort with any positive probability. Apriori, it is reasonable to expect that there would be mixed-strategy equilibria where agents randomize between working and shirking, \ie, put in effort (choose $e_{ij} =1$) sometimes and not (choose $e_{ij} =0$) some other times. However, we next show that as long as tasks are apriori equivalent and agents only randomize between reporting truthfully and reporting the outcome of an independent random cointoss (\ie, they do not invert evaluations), the {\em only} equilibrium in which any agent uses any support on $(0,r)$ is the one in which {\em all} agents {\em always} use $(0,r)$ on all their tasks. To show this, we start with the following useful lemma saying that an agent who uses a low-effort strategy any fraction of the time will always have a beneficial deviation as long as some reference agent plays $(1, \truth)$ with some positive probability. % (assuming all agents mix only between $(1, \truth)$ and $(0,r)$ strategies, \ie, excluding $(1, \invert)$).
Roughly speaking, this is because as long as there is some probability that an agent's reference rater plays $(1, \truth)$ rather than $(0,r)$, the agent strictly benefits by always playing $(1, \truth)$ to maximize the probability of both agents playing $(1, \truth)$, which is the only time the agent obtains a positive reward.

\begin{lemma}
\label{l-deviate}
Suppose the probability of agent $i$ using strategy $(1,\truth)$ is $\delta$ and strategy $(0,r_i)$ is $1-\delta$ for each task $j \in J(i)$.  Suppose $i$'s potential reference raters $\ipj$ use strategies $(1, \truth)$ and $(0, r_{\ipj})$ with probabilities $\eps_{\ipj}$ and $1-\eps_{\ipj}$ respectively, for each task $j \in J(i)$. If $\eps_{\ipj} > 0$ for any reference rater with proficiency $p_{\ipj} > 1/2$, then agent $i$ has a (strict) profitable deviation to $\delta' = 1$, \ie, to always using strategy $(1, \truth)$, for all values of $r_i \in [0,1]$.
\end{lemma}

\begin{proof}
Consider a particular task $j$, and let $k = 1, \ldots, K$ be the potential reference rater for $(i,j)$. Let $a_k$ denote the probability that $k$ is the reference rater for agent $i$ for task $j$. By linearity of expectation, $i$'s expected reward for $j$ can be written as
{\small \begin{align*}
E[R_{ij}]&= \sum_{k=1}^K a_{k}\Big[\delta\eps_{k}(p_i p_k + (1-p_i)(1- p_k) - (p_i[H] p_k[H] + (1-p_i[H])(1-p_k))) \\
&\quad+~ (1-\delta)\eps_k(r_ip_k[H] + (1-r_i)(1- p_k[H]) - (r_ip_k[H] + (1-r_i)(1-p_k))) \\
&\quad+~ \delta(1-\eps_k)(p_i[H] r_k + (1-p_i[H])(1- r_k) - (p_i[H]r_k + (1-p_i[H])(1-r_k))) \\
&\quad+~ (1-\delta)(1-\eps_k)(r_i r_k + (1-r_i)(1- r_k) - (r_i r_k + (1-r_i)(1-r_k)))\Big] \\
&= \delta \sum_{k} a_k\eps_k(p_i p_k + (1-p_i)(1- p_k) - (p_i[H] p_k[H] + (1-p_i[H])(1-p_k[H]))) \\
&= \delta \sum_{k} a_k\eps_k f_1(p_i, p_k).
\end{align*}}
Now, $E[R_{ij}]$ is linear in $\delta$, and by Lemma \ref{l-fpq}, the coefficient of $\delta$ is nonnegative for all $\eps_k$ and $p_k \geq 1/2$, and strictly greater than $0$ if $\eps_k > 0$ for some $k$ with $p_k > 1/2$. Therefore, $i$ can strictly increase her expected reward $E[R_{ij}]$ by increasing $\delta$ for any $\delta < 1$, as long as there is some reference agent $k$ with $\eps_k > 0$ and $p_k > 1/2$.

The same argument holds for each task $j \in J(i)$, and therefore to strictly improve $i$'s total reward $E[R_{i}]$, we only need one reference rater across all tasks to satisfy $\eps_k > 0$ and $p_k > 1/2$ to obtain a strictly beneficial deviation (recall that we assumed $p_i \geq 1/2$ for all $i$).
\end{proof}

This lemma immediately allows us to show that the only low-effort equilibria of $\M$ that we reasonably\footnote{(We say reasonably because of the technical possibility of equilibria where some agents mix over $(1, \invert)$ as well.)} need to be concerned about is the pure-strategy equilibrium in which $e_{ij}=0$ for all $i,j$. Note that different agents could use different $r_i$ (or even $r_{ij}$) in such equilibria, but all agents will receive reward $0$ in {\em all} such equilibria.

\begin{theorem}
Suppose every agent can be a reference rater with some non-zero probability for every other agent, and tasks are apriori equivalent. Then, the only equilibria (symmetric or asymmetric) in which agents mix between $(1, \truth)$ and any low-effort strategy $[(0, r_{ij})]$ with non-trivial support on $[(0, r_{ij})]$ are those where all agents always use low effort on all tasks.
\end{theorem}

%\begin{proof}
%The only symmetric pure-strategy equilibria are $(0,r)$ for any $r \in [0,1]$, and $(p_i, \truth)$ (as well as $(p_i, \invert)$, which has exactly the same reward as $(p_i, \truth)$. The reward to each agent in any $(0,r)$ equilibrium is $0$, whereas the reward to each agent in $(p_i, \truth)$ is positive when $p_i > 1/2$ by Lemma \ref{l-fpq}, and therefore the maximum reward.
%\end{proof}

\noindent{\bf Eliminating low-effort equilibria.} Our final result uses Lemma \ref{l-deviate} to obtain a result about eliminating low-effort equilibria. Suppose there are some trusted agents (for example, an instructor or TA in the peer-grading context or workers with long histories of accurate evaluations or good performance in crowdsourcing platforms) who always report truthfully with proficiency $t > 1/2$. Let $\eps_t$ denote the minimum probability, over all agents $i$, that the reference rater for agent $i$ is such a trusted agent (note that we can ensure $\eps_t >0$ by having the trusted agent randomly choose each task with positive probability). Lemma \ref{l-deviate} immediately gives us the following result for $\M_d$, arising from the fact that the reward from playing a random strategy $(0,r)$ is exactly $0$---the presence of trusted agents with a non-zero probability, however small, is enough to eliminate low-effort equilibria altogether. %This is because if there is any positive probability that an agent's reference rater is sure to put in effort and truthfully report her evaluation, the agent can profitably deviate away from playing $(0,r)$, as in the previous result. We state this formally below.
\begin{theorem}
\label{t-norand}
Suppose $\eps_t > 0$. Then $[(0,r_{ij})]$ is not an equilibrium of $\M$ for any $r_{ij} \in [0,1]$.
\end{theorem}
\begin{proof}
Suppose all agents except the trusted agent use the strategy $(0,r_{ij})$, and $\eps_t$ is the probability that the trusted agent is the reference rater for any agent-task pair. Then, since agent $i$ reports $X_{ij}$ according to a random coin toss independent of the actual realization of $j$, the payoff from any reference rater, whether the trusted agent or another agent playing $(0,r)$ is $0$. For notational simplicity, let $r = r_{ij}$, $r'=r_{\ipj j}$.
{\small \begin{align*}
E[R_{ij}]& = \eps_t(rt[H] + (1-r)(1-t[H])  - (rt[H] + (1-r)(1-t[H])) \\
&\quad +~ (1-\eps_t)(r r' + (1-r)(1-r') - (rr' + (1-r)(1-r')))\\
&=0.
\end{align*} }
By deviating to $(1,\truth)$, agent $i$ can strictly improve her payoff as long as $\eps_t > 0$ and $t, p > 1/2$, since her expected reward from this deviation is
{\small \begin{align*}
E[R_{ij}] &= \eps_t(pt + (1-p)(1-t) - (\sph t[H] + (1-\sph)(1-t[H])) \\
&\quad + ~(1-\eps_t)(r r' + (1-r)(1-r') - (rr' + (1-r)(1-r'))\\
&> 0,
\end{align*} }
since the coefficient of $\eps_t$ is positive for $t, p > 1/2$ by Lemma \ref{l-fpq}.
Therefore, there is a strictly beneficial deviation to $(1, \truth)$, so there is a choice of multiplier for the reward such that the payoff to agent $i$, which is the difference between the reward and the cost of effort $c$, is strictly positive as well. So $(0,r_{ij})$ is not an equilibrium of $\M$ when $\eps_t > 0$.

%Suppose all agents except the trusted agent use the strategy $(0,r)$, and $\eps_t$ is the probability that the trusted agent is the reference rater for any agent-task pair. Then, since agent $i$ reports $X_{ij}$ according to a random coin toss independent of the actual realization of $j$, the payoff from any reference rater, whether the trusted agent or another agent playing $(0,r)$ is $0$:
%\begin{align*}
%E[R_{ij}] &= \eps_t(rt[H] + (1-r)(1-t[H]) - (rt[H] + (1-r)(1-t[H])) + (1-\eps_t)(r^2 + (1-r)^2 - (r^2 + (1-r)^2))\\
%&=0.
%\end{align*}
%By deviating to $(1,\truth)$, agent $i$ can strictly improve her payoff as long as $\eps_t > 0$ and $t, p > 1/2$, since her expected reward from this deviation is \begin{align*}
%E[R_{ij}] &= \eps_t(pt + (1-p)(1-t) - (\sph t[H] + (1-\sph)(1-t[H])) + (1-\eps_t)(r^2 + (1-r)^2 - (r^2 + (1-r)^2))\\
%&> 0,
%\end{align*}
%since the coefficient of $\eps_t$ is positive for $t, p > 1/2$ by Lemma \ref{l-fpq}.
%Therefore, there is a strictly beneficial deviation to $(1, \truth)$, so there is a choice of multiplier for the reward such that the payoff to agent $i$, which is the difference between the reward and the cost of effort $c$, is strictly positive as well. So $(0,r)$ is not an equilibrium of $\M$ when $\eps_t > 0$.
\end{proof}

This result, while simple, is fairly strong: as long as some positive fraction of the population can be trusted to always report truthfully with proficiency greater than $1/2$, the only reasonable\footnote{Again, we say reasonable rather than unique because $(1, \invert)$ does remain an equilibrium of $\M$ for all $\eps_t$ less than a threshold value--- however, in addition to being an unnatural and risky strategy, this equilibrium yields strictly smaller payoffs than $[(1, \truth)]$ when $\eps_t > 0$. Note also that the introduction of such trusted agents does not introduce new equilibria, and that $[(1, \truth)]$ remains an equilibrium of $\M$.)}
 equilibrium of $\M$ is the high-effort equilibrium $[(1, \truth)]$, no matter {\em how small} this fraction. In particular, note that $\M$ does not need to assign a higher reward for agreement with a trusted agent to achieve this result, and therefore {\em does not need to know the identity} of the trusted agents. In contrast, the mechanism which rewards agents for agreement with a reference rater without subtracting out our statistic term must use a higher reward $w(\eps_t)$ for agreement with the trusted agents which increases as $\frac{1}{\eps_t}$ to eliminate low-effort equilibria\footnote{The same is the case for a mechanism based on rewarding for the `right' variance, which does retain $[(1,X)]$ as a maximum reward equilibrium, but still requires identifying the trusted agents and rewarding extra for agreement with them.}--- this, in addition to being undesirably large, also requires identification of trusted agents.

%\iffalse
\section{Creating the Task Assignment}
\label{s-graph}
%\input graph
%\section{Creating the Task Assignment}
%\label{s-graph}
While in some crowdsourcing settings, agents choose tasks at will, there are also applications where a principal can potentially choose an assignment of a collection of her tasks among some assembled pool of workers. In this section, we present a simple algorithm to design assignment of tasks to agents such we can satisfy the condition in Theorem~\ref{t-eqbm} for mechanism $\M_{D-1}$, \ie, when $d = D-1$. We note that with this assignment of tasks to agents, choosing reference raters appropriately is trivially feasible for $d = 1$, \ie, for $\M_1$, and ensuring $d_{ij} = d$ is also easy as described in \S \ref{s-mech}.

We start out by randomly permuting all agents using a permutation $\pi$. For simplicity of presentation we assume that $\frac{m}{D}$ ($=\frac{n}{T}$) is an integer. The $m$ tasks are divided into $\frac{m}{D}$ task-blocks, each containing $D$ tasks. Similarly, the $n$ agents are divided into $T$ agent blocks, each containing $\frac{n}{T}$ agents. We number the task-blocks by $b=1,\ldots,\frac{m}{D}$ and the agent blocks by $a=1,\ldots,T$. The agents in block $a$ are thus $(a-1)\frac{n}{T} + 1,\ldots,a\frac{n}{T}$ and the tasks in block $b$ are $(b-1)D+1,\ldots,bD$.
%Let the {\em residual capacity} of a task denote the number of agents that still need to be assigned to it so that the total number of agents assigned is $C$. Similarly, the residual capacity of an agent is the total number of tasks that can still be assigned to it.

We first describe the algorithm and then show that it produces an assignment that satisfies the conditions required in the definition of $\M_{D-1}$, in particular that for each agent-task pair, it is possible to choose a reference rater who has only that task in common with this agent. The algorithm works as follows: we assign tasks for agents starting from the agent block $a=1$ onwards. For block $1$, each agent $i'$ in the block is assigned all the tasks corresponding to the task block $i'$ (recall that number of agents in a block equals $\frac{n}{T} = \frac{m}{D}$, the number of task-blocks). This completes fills up the capacity of the agents in block $1$. For blocks $a = 2,\ldots, T$, consecutively, the agent $(a - 1)\frac{n}{T} + i'$ is assigned $D$ tasks $\{i', i' + \frac{m}{D},\ldots, i' + \frac{m}{D}(D-1)\}$, for $i' =1,\ldots,\frac{n}{T}$.

The above assignment completely describes the sets $J(i)$ and $I(j)$ for every agent $i$ and task $j$. For each task $j$, let $i^*_j$ denote the unique agent in block $1$ who works on task $j$. We define the reference raters as follows: for each agent-task pair $(i, j)$, if $i$ lies in blocks $\{2,\ldots,T\}$,  define the reference rater $\ipj  = i^*_j$. If $i$ lies in block $1$, define the reference rater to be any other user who is working on this task. Note that for $d= D-1$, the sets $S_{ij}$ and $S_{\ipj  j}$ are exactly $S_{ij} = J(i) \setminus \{j\}$ and $S_{\ipj  j}=J(\ipj) \setminus \{j\}$.

The following lemma proves two things--- first, the assignment above is actually feasible under fairly mild conditions, and second, that the choice of reference raters satisfies the conditions in the definition of $\M$ and those required by Theorem \ref{t-eqbm}.

\begin{lemma}
If $m \ge D^2$, the above algorithm generates a feasible assignment, i.e. every agent is assigned exactly $D$ tasks and every task to $T$ agents. Also, for agent-task pair $(i,j)$, the reference rater $\ipj$ satisfies $J(\ipj ) \cap J(i) = \{j\}$. Furthermore, $E[p_{\ipj }] = E_{j_l \in T_{ij}}[p_{i'_{j_l}}]$.
\end{lemma}
\begin{proof}
Agents in block $1$ are clearly assigned to their full capacity. For blocks $a = 2, \ldots, T$, for every agent $i = (a - 1)\frac{n}{T} + i'$, the set $I(i) = \{i', i' + \frac{m}{D},\ldots, i' + \frac{m}{D}(D-1)\}$. Note that for each $i$, the above values are all distinct, and that $i' + \frac{m}{D}(D-1) \leq \frac{n}{T} + \frac{m}{D}(D-1) = m$. Thus every agent's assignment is feasible. Since the total capacity of agents equals the total capacity of the tasks, the tasks are also assigned completely, and to distinct agents.

In order to see that the choice of reference raters is feasible, note that if $D^2 \le m$, then $\frac{m}{D} \ge D$, and hence the tasks for each agent belong to distinct blocks. For agent-task pairs $(i, j)$ where the agents are in blocks $2,\ldots,T$, the reference rater $\ipj = i^*_j$, the unique agent in block $1$ who worked only on the task-block that $j$ belongs to. By the above argument, $i$ does not work on any other task from this block, and hence $J(i^*_j) \cap J(i) = \{j\}$. By the same argument, $i$ is also a feasible reference rater for $i^*_j$ on task $j$. Thus, the choice of reference raters satisfies the condition for $\M_{D-1}$.

Finally, the expectation condition follows simply from the random permutation applied to the set of agents at the beginning of the construction.

\end{proof}

\section{Discussion}
\label{s-disc}
In this paper, we introduced the problem of information elicitation when agents' proficiencies are endogenously determined as a function of their effort, and presented a simple mechanism which uses the presence of multiple tasks to identify and penalize low-effort agreement to incentivize effort when tasks have binary types. Our mechanism has the property that maximum effort followed by truthful reporting is the Nash equilibrium with maximum payoff to all agents, including mixed strategy equilibria. In addition to handling endogenous agent proficiencies, to the best of our knowledge this is the first mechanism for information elicitation with this 'best Nash equilibrium' property over all pure and over mixed strategy equilibria that  %(\ie, that (maximum effort)-truthtelling is the best Nash equilibrium)
%that does not require information about priors to compute rewards,
requires agents to only report their own evaluations (\ie, without requiring `prediction' reports of their beliefs about other agents' reports), and does not impose any requirement on a diverging number of agent reports per task to achieve its incentive properties.  Our mechanism provides a starting point for designing information elicitation mechanisms for several crowdsourcing settings where proficiency is an endogenous, effort-dependent choice, such as image labeling, tagging, and peer grading in online education.

We use the simplest possible model that captures the complexities arising from strategically determined agent proficiencies, leading to a number of immediate directions for further work. First, our underlying outcome space is binary ($H$ or $L$)--- modeling and extending the mechanism to allow a richer space of outcomes and feedback is one of the most immediate and challenging directions for further work. Also, our model of effort is binary, where agents either exert full effort and achieve maximum proficiency, or exert no effort to achieve the baseline proficiency. While our results extend to a model where proficiency increases linearly with cost, a natural question is how they extend to more general models, for example, with convex costs.  Finally, a very interesting direction is that of {\em heterogenous} tasks with task-specific priors and abilities. In our model, tasks are homogenous with the same prior $\ph$, and agents have the same cost and maximum proficiency for each task. If tasks differ in difficulty, and agents can observe the difficulty of a task prior to putting in effort, there are clear incentives to shirk on harder tasks while putting in effort for the easier ones. While tasks are indeed apriori homogenous (or can be partitioned to be so) in some crowdsourcing settings, there are other applications where some tasks are clearly harder than others; also, agents may have task-specific abilities. Designing mechanisms with strong incentive properties for this setting is a very promising and important direction for further work. \\
%AG: Moved to footnote in mech. We note finally that understanding the effect of the parameter $d$ in our mechanisms, which appears irrelevant to the mechanism's behavior when agents re risk-neutral agents, for risk-averse agents is also an interesting open question.

%\begin{itemize}
%\item check continuous effort, linear $p$, non-linear $p$ (arbitrary concave equiv. convex costs).
%\item $H/L$ to richer space of underlying ground truths.
%\item different tasks are different. This might also mean that assignment graph may not be completely designed--- agents might have incentives to shirk on specific tasks. This is an important point if we want to extend this to all crowdsourcing.
%\item distinguishing between $c=1$ and $c=C-1$ mechanisms. For risk averse agents these mechanisms might have different properties.
%\end{itemize}
{\bf Acknowledgements.} We thank David Evans, Patrick Hummel and David Stavens for helpful discussions and pointers to related literature, and anonymous referees for comments and suggestions that helped improve the presentation of the paper.
\bibliographystyle{abbrv}
\bibliography{peer}

\end{document}